\documentclass[12pt]{article}

\usepackage{amssymb,amsmath,amsfonts,eurosym,geometry,ulem,graphicx,caption,color,setspace,sectsty,comment,footmisc,caption,natbib,pdflscape,subfigure,array}
\usepackage{bbm}
\usepackage{color}  
\usepackage[hidelinks,colorlinks=true,linkcolor=blue,citecolor=blue]{hyperref}

\usepackage{natbib}

\newcolumntype{L}[1]{>{\raggedright\let\newline\\arraybackslash\hspace{0pt}}m{#1}}
\newcolumntype{C}[1]{>{\centering\let\newline\\arraybackslash\hspace{0pt}}m{#1}}
\newcolumntype{R}[1]{>{\raggedleft\let\newline\\arraybackslash\hspace{0pt}}m{#1}}

\usepackage[T1]{fontenc}
\usepackage[utf8]{inputenc}
\usepackage{setspace}
\usepackage{babel}
\usepackage{blindtext}
\usepackage[utf8]{inputenc}
\usepackage{amssymb}
\usepackage{amsthm}
\usepackage{amsmath}
\newtheorem{theorem}{Theorem}
\newtheorem{corollary}{Corollary}

\newtheorem{proposition}{Proposition}

\newtheorem{remark}{Remark}

\newtheorem{manualtheoreminner}{Theorem}
\newenvironment{manualtheorem}[1]{%
  \renewcommand\themanualtheoreminner{#1}%
  \manualtheoreminner
}{\endmanualtheoreminner}

\usepackage{array}
\usepackage{footmisc}

\newtheorem{definition}{Definition}
\DeclareMathOperator*{\argmax}{arg\,max}
\DeclareMathOperator*{\argmin}{arg\,min}

\title{Auctioning Multiple Goods without Priors}
\author{Wanchang Zhang\thanks{Department of Economics, University of California, San Diego.  Email: waz024@ucsd.edu. I am indebted to    Songzi Du and Joel Sobel for stimulating discussions. I thank  Denis Shishkin and Emanuel Vespa for helpful comments.}
}
\date{Current Draft: April 2022\\
First Draft: October 2021}

\begin{document}

\maketitle
\begin{abstract}
 I consider a mechanism design problem of selling multiple goods to multiple bidders when the  designer has minimal amount of information. I assume that the   designer only knows the upper bounds of bidders' values for each good  and has no additional distributional information. The  designer takes a minimax regret approach. The expected regret from a mechanism given a joint distribution over value profiles and an equilibrium is defined as the difference between the full surplus and the expected revenue.    The designer seeks a mechanism, referred to as a \textit{minimax regret mechanism},  that minimizes her worst-case expected regret across all possible joint distributions over value profiles and all  equilibria. I find that a \textit{separate second-price auction with random reserves} is a minimax regret mechanism for general upper bounds. Under this mechanism, the  designer holds a separate auction for each good; the formats of these auctions are second-price auctions with random reserves.
\vspace{0in}
\end{abstract}
\newpage
\section{Introduction}
 The  standard auction literature  focuses on the single-good environment and  assumes that bidders' value profile follows a commonly known joint distribution. It is assumed that the designer  seeks a mechanism that maximizes the expected revenue. \cite{myerson1981optimal} characterizes optimal mechanisms for selling a single good when bidders' values are independent; \cite{cremer1988full} characterize optimal mechanisms for selling a single good given generic correlation structures of bidders' values. However,  optimal mechanisms vary widely with the model of correlation structure and relatively little is known about how optimal mechanisms would perform once the correlation structure is misspecified. In addition, it is not clear how the designer should form  a prior in the first place. \\ 
 \indent In this paper, I am going to extend the analysis in two ways. First, I consider the  multiple-good environment, in which little is known about the optimal mechanisms. Even for the special case where there is only one bidder, the optimal mechanism is hard to characterize or to describe (\cite{daskalakis2014complexity} and \cite{manelli2007multidimensional}).  Second, I consider a robust version of the analysis.  Specifically,  I consider a (correlated) private value model where the designer auctioning multiple different goods   knows no distributional information except for the \textit{upper bounds} of bidders' values for each good. In contrast, bidders agree on a joint distribution over their value profiles\footnote{In the Appendix, I show that the main result still holds when  bidders can acquire additional information.}.  The designer  considers any joint distribution consistent with the known upper bounds to be possible. I consider general mechanisms with the
only requirement that the mechanism ``secures'' bidders' participation: there exists a message for each payoff type of each bidder that guarantees a non-negative ex-post payoff regardless of the other bidders' messages.  I assume that the designer takes the \textit{minimax regret approach}. Precisely, the \textit{expected regret} from a mechanism given a joint distribution over value profiles and a Bayesian equilibrium is defined as the difference between the full surplus\footnote{The full surplus under a joint distribution over value profiles is the expected revenue attainable were the designer able to sell the goods with full information about bidders' value profile.} and the expected revenue. The designer evaluates a mechanism by its highest expected regret across all possible joint distributions and all Bayesian equilibria, which is referred to as its \textit{regret cap}. The designer aims to find a mechanism, referred to as a \textit{minimax regret mechanism},  that minimizes the regret cap.\\
\indent The assumption that the designer only knows the upper bounds  is appropriate for situations where little information is known about the bidders and it is costly and time-consuming to collect  information. For example, in an auction of initial public offerings, there is no distributional information about the bidders' values. Note that  in this example,   bidders' budgets, which can be viewed as (reasonable approximations of) the upper bounds of bidders' values,    are typically known by the designer, as  bidders for initial public offerings are often institutional investors, whose financial
resources are publicly known, or can be estimated fairly precisely from their financial reports. On the contrary, the assumption may be too conservative for situations in which data about bidders are abundant, e.g., online advertising in which auctions are held repeatedly and frequently.  In addition,  the assumption is  formally necessary to obtain non-trivial results because if there is no known upper bound, then minimax regret will be infinite \footnote{See Remark \ref{r5} for the formal proof.}. Thus, this model can be viewed as a theoretical benchmark that provides a first step toward a broad study of robust auction design problems in the multiple-good environment.\\
\indent The minimax regret approach can  be traced back to \cite{wald1950statistical} and \cite{savage1951theory}. It captures the idea that a decision maker is concerned about missing out opportunities. A decision theoretical axiomatization of regret  can be found in \cite{milnor1951games} and \cite{stoye2011axioms}.   It is adapted to multiple priors by \cite{hayashi2008regret} and \cite{stoye2011axioms}. Another leading approach is the maxmin utility approach, which is adopted by most of the robust mechanism design literature.  However, in the setting of this paper,  under the maxmin utility approach, it is optimal for the designer to keep all goods to herself because  it is possible that all bidders have zero values towards all goods. Note that in this extreme case, there is no surplus to  extract even under complete information.  Thus, the maxmin utility approach is too conservative to be useful, whereas the minimax regret approach  protects the surplus when there is some surplus to extract and will be shown to lead to a non-trivial answer. \\
\indent The main result is that a \textit{separate second-price auction with random reserves} is a minimax regret mechanism. This mechanism can be described as follows. For each good, the designer holds a separate auction; the formats of  these auctions are second-price auctions with bidder-specific random reserves that depend on the upper bounds of  values\footnote{The distributions of these random reserves are given in the Section \ref{s51}. }. It is remarkable that a simple mechanism arises as a robustly optimal mechanism for auctioning multiple goods, across   all participation-securing mechanisms that include highly complicated mechanisms, e.g., combinatorial auctions\footnote{In a combinatorial auction, bidders can place bids on combinations of discrete heterogeneous goods.}.  \\
\indent Importantly, I allow for general upper bounds of values for the main result. In particular, the upper bounds of the values for a given good can be different across bidders. This captures the widely observed \textit{asymmetries} in many real-life auction environments. For example, in art auctions,  there are obvious asymmetries associated with differing budget constraints across bidders. Asymmetric auctions have been studied in \cite{maskin2000asymmetric}, \cite{hafalir2008asymmetric},\cite{guth2005bidding} and \cite{athey2013set} among others. \\
\indent The main result provides a possible explanation why  separate second-price auctions - or their more practical equivalents in the private value environment, separate English auctions - are widely used  in practice for auctioning multiple goods. For example, at the popular online auction site eBay, each good is typically auctioned separately via an English auction (\cite{krishna2009auction}, \cite{anwar2006bidding} and  \cite{feldman2020simultaneous}); Sotheby’s
and Christie’s (two major auction houses of art) typically sell works of art separately via an English auction (\cite{ashenfelter2003auctions}).\footnote{In practice, there are other mechanisms used for auctioning multiple goods. While most of spectrum auctions in US do not allow for  ``package bidding'' due to its complexities (\cite{cramton2002spectrum} and \cite{filiz2015multi}),  the Federal Communications Commission (FCC) has used a ``package bidding auction'' to sell spectrum licenses in rare cases. In such an auction, a bidder is allowed to select a group of licenses to bid on as a package. In an early 2008 FCC auction (Auction 73), AT\&T and Verizon both bought geographically diverse packages of 700MHz spectrum. The principal rationale to consider a package bidding auction is that there may be complementarity between different licenses (\cite{goeree2005comparing}).    Complementarity is ruled out in my model. It is an interesting open question what the minimax regret mechanism would be like when there is complementarity between goods.}   The main  result  justifies this empirical rule of thumb by an optimal performance guarantee: a separate second-price auction (albeit with random reserves) minimizes the worst-case expected regret. This may be one reason why  complicated mechanisms that require the full information of the joint distribution over bidders' value profiles, to my knowledge,  are not used in practice for auctioning multiple goods. \\
\indent The role of randomized reserves  can be   seen considering the one-good one-bidder case, in which they are reduced to randomized pricing. The designer suffers from a large regret if she charges a high price when the value of the bidder is low or if she charges a low price when the value of the bidder is high.  She can lower her regret by randomizing. \cite{bergemann2008pricing}  characterize the solution for the one-good one-bidder case. Indeed, the well-crafted distribution of the randomized pricing renders the designer indifferent across  values over a range.  The second-price auction with random reserves extends the robust property to the one-good  multiple-bidder case. In this case, the regret from a value profile is the difference between the highest value and the collected revenue. When the second highest value is low enough (e.g., 0),  it boils down to the one-good one-bidder case and the regret remains the same; when the second highest value is high enough (e.g., above the lower bound of the random reserve for the highest bidder), then the revenue is even higher and the regret is thus lower. To see the intuition behind separation, it is instructive to consider another mechanism of  auctioning the  \textit{grand} bundle (only the bundle of all goods is auctioned). I argue that this mechanism may result in a high regret. Consider a three-good three-bidder example and an extremely \textit{asymmetric} value profile in which each bidder values a different good (assuming that the upper bound on each bidder's values for each good is 1) : $(v_1^1,v_1^2,v_1^3)= (1,0,0), (v_2^1,v_2^2,v_2^3)= (0,1,0), (v_3^1,v_3^2,v_3^3)=(0,0,1)$.\footnote{The superscript represents the good, and the subscript represents the bidder.} The designer will \textit{lose all but one good} if auctioning the grand bundle: she can at most obtain a revenue of 1 from one of the goods but will suffer a regret of 2 from losing the other goods. In contrast, separation can guarantee a good regret performance for each good. It can be shown that for this example,  the separate second-price auction with random reserves  yields a regret cap\footnote{Indeed, it can be shown that the regret cap is $\frac{3}{e}$ by simple calculation.} lower than 2.  Intuitively, auctioning the grand bundle performs just like selling one good at this value profile, while selling separately allows the designer to earn more. Furthermore, the same argument implies that \textit{partial} bundling (a mechanism in which a bundle of some goods are auctioned) may perform worse than separate selling.   \\
\indent I show that the separate second-price auction with random reserves is a minimax regret mechanism by constructing a  joint distribution over value profiles, referred to as a  \textit{worst-case  distribution}, such that the lower bound of the expected regret for any mechanism and any equilibrium under this joint distribution is \textit{equal} to  the upper bound of the expected regret for any joint distribution and any equilibrium under the separate second-price auction with random reserves. One can imagine that adversarial nature  is constructing a worst-case distribution to let the designer suffer from a high expected regret.\\
\indent The worst-case  distribution admits a simple description as follows\footnote{Its formal definition  is given in Section \ref{s52}. }. For each good, adversarial nature \textit{selects} one bidder whose upper bound of the values of the good is the highest among the bidders (breaking ties arbitrarily).  For each bidder, the marginal distributions of the values of the goods for which the bidder is selected are \textit{equal-revenue distributions}, defined by the property of   a unit-elastic demand: in the monopoly pricing problem, the monopoly's revenue from charging any price in the support  is the same; the values for these goods are \textit{comonotonic} (maximal positive correlation); the values for the other goods are all  \textit{zeros}\footnote{Note that if a bidder is not selected for any good, then his values for all goods are zeros.}. For the goods across the  bidders, the values are \textit{independent}.\\
\indent Under this distribution,  each selected bidder 
values  a totally different set of goods; for each good,   exactly one bidder values it; each selected bidder's values for the goods he values are comonotonic; the values of the goods across the selected bidders are independent.\\
\indent Now I illustrate the idea behind this  distribution. First,  to understand the part of equal-revenue distributions, consider the  one-good one-bidder case in which the mechanism collapses to randomized pricing over a range. As the designer is indifferent between these prices, the marginal revenue must be zero over these prices, from which  an equal-revenue distribution arises\footnote{See  \cite{bergemann2008pricing} for details about the derivation.}. Second, the intuition for the part of  selection can be summarized by a \textit{scale effect}: because the minimax regret in the one-good one-bidder case is proportional to the upper bound of the values, by selecting a bidder whose upper bound of the values is the highest for each good,  the potential regret is made  the highest for each good. Third, the intuition for the part of comonotonicity  can be summarized by a \textit{screening effect}: consider the  multiple-good one-bidder case,  the comonotonicity between goods limits the ability of the designer to screen different goods by reducing the multi-dimensional screening to the single-dimensional screening.   Fourth, the intuition for the part of  zeros  can be summarized by a \textit{competition effect}: it eliminates the competition among bidders for each good by letting  only one bidder have a positive value for each good. Fifth, the intuition for the part of independence can be summarized by an \textit{information effect}: one bidder's  values for goods  do not provide any information about any other bidder's  values for  other goods, which prevents the designer from extracting 
surplus from  one bidder based on information about  other bidders.  \\
\indent The main  result incorporates multi-dimensional screening and single-good auction as two special cases.  For the multi-dimensional screening, a  \textit{separate randomized posted-price mechanism} is a minimax regret mechanism, and  a  distribution in which the  values across the goods are comonotonic is a worst-case  distribution.\footnote{The solution for the multi-dimensional screening has been found by \cite{koccyiugit2021robust}. I offer an alternative proof using a quantile-version of virtual values. \cite{carroll2017robustness} also uses quantiles to parameterize the single buyer's values for the multi-dimensional screening.} For the single-good auction, a \textit{second-price auction with random reserves} is a minimax regret mechanism, and a distribution in which only one bidder has a positive value for the good is a worst-case  distribution. \\
\indent The remainder of the introduction discusses related work. Section \ref{s3} presents the model.  Section \ref{s4} illustrates the methodology. Section \ref{s5} characterizes the main result.   Section \ref{s6} presents the solutions to two special cases.  Section \ref{s7} is a discussion.  Section \ref{s8} is a conclusion. The Appendix extends the result to an environment where bidders can acquire any additional information.
\subsection{Related Work}\label{s2}
The closest related work is \cite{koccyiugit2020regret}, who consider the same environment and find a (different) separate second-price auction with random reserves has good robust properties. There are, however, several critical differences. First, they restrict attention to dominant-strategy  mechanisms\footnote{See Definition \ref{d1} for the formal definition of dominant-strategy mechanisms.}, whereas I allow for general mechanisms with  essentially the only requirement that there is a message that secures  bidders' participation. That is, I search for a minimax regret mechanism from  a much wider class of mechanisms.   Second,  they show that their proposed mechanism is a minimax regret mechanism for the symmetric case where  the upper bounds of the values for a given good are the same across bidders, whereas I establish that my proposed mechanism is a minimax regret mechanism for general upper bounds. That is, I place no restrictions on the upper bounds of the  values for a given good across bidders. It is  important and interesting to understand the minimax regret mechanism in asymmetric environments considered in this paper, as the symmetric case is a knife-edge case. In this sense, this paper complements their work.  The key factor that drives these differences is that I construct a different  joint distribution over value profiles that yields a higher lower bound of the expected regret in general. In addition, there is a Pareto ranking between my proposed mechanism and theirs:  in the truth-telling equilibrium, the ex-post regret is always weakly lower and sometimes strictly lower under my proposed mechanism than that under theirs (Remark \ref{r2}).  Besides, technically,  they take the duality approach for their result, whereas I adopt an adaptation of the classic Myerson's approach to identify a lower bound of the expected regret under my constructed joint distribution.\\
\indent \cite{bergemann2008pricing,bergemann2011robust}  consider the problem of monopoly pricing where the monopolist is faced with uncertainty about the demand curve and characterize randomized posted price mechanisms as minimax regret mechanisms. My result is reduced to that of  \cite{bergemann2008pricing} in the  one-good one-bidder case. \cite{koccyiugit2021robust} consider the problem of multi-dimensional screening without priors and characterize randomized separate posted price mechanisms as minimax regret mechanisms. My result is reduced to theirs in the multiple-good one-bidder case. Moreover, I offer another minimax regret mechanism: a randomized grand bundling (Remark \ref{r7}). \\
\indent There are other  mechanism design papers where the designer aims to minimize the worst-case regret. \cite{guo2019robust} study the classic problem of monopoly regulation (\cite{baron1982regulating}) using the minimax regret approach; \cite{guo2021project} study a project choice problem (\cite{armstrong2010model}) using the minimax regret approach; \cite{caldentey2017intertemporal} characterize the deterministic dynamic pricing rule that minimizes the seller's worst-case regret; \cite{renou2011implementation} study the problem of implementing social
choice correspondences using $\epsilon-$minimax regret criterion.\\
\indent More broadly, this paper is related to the robust mechanism design literature and the information design literature. \cite{carrasco2018optimal} characterize maxmin selling mechanisms  when the seller faced with a single buyer only knows the first $N$ moments of distribution ($N$ is an arbitrary positive integer). \cite{che2019distributionally}, \cite{he2022correlation} and \cite{zhang2021correlation} study the robust auction design problem when the designer has limited distributional information.  \cite{zhang2021robust} studies the profit-maximizing bilateral trade problem and characterize maxmin trade mechanisms when the designer knows only the expectations of the values.  Similar to mine, these papers all assume that  the values are private and all characterize some randomized mechanism as a maxmin mechanism. \cite{carroll2017robustness} studies the multi-dimensional screening problem when the designer only knows the marginal distributions. Similar to the multiple-good one-bidder case in my paper, a separate selling mechanism turns out to be a maxmin solution.  Different from the multiple-good one-bidder case in my paper, his maxmin solution does not require randomization. \cite{chung2007foundations} and \cite{chen2018revisiting} study  maxmin foundations for dominant-strategy mechanisms. Similar to my model, they study the private value environment. In contrast to my model, beliefs are not required to be consistent with a common prior; in addition, they select for the designer's most preferred equilibrium. \cite{roesler2017buyer} and \cite{condorelli2020information} derive optimal information structures for maximizing buyer's surplus. My worst-case distribution is reduced to theirs in the  one-good  one-bidder case. \cite{du2018robust} derives the optimal informationally robust mechanism for the  one-good  one-bidder case and constructs a mechanism that asymptotically extracts full surplus for the single-good auction.  \cite{brooks2021optimal} derive the optimal  informationally robust mechanism for the single-good auction. In contrast to my model, they study the common value environment. However, our solution concepts are similar. My solution is indeed a \textit{strong minimax solution}:  holding the joint distribution 
fixed, the mechanism and equilibrium minimize  regret, and holding the mechanism fixed,
the joint distribution and equilibrium maximize regret; in addition, there is an equilibrium (the truth-telling equilibrium) under which the regret cap is hit.

\section{Model}\label{s3}
I consider a (correlated) private value environment where a designer sells $J$ different indivisible  goods to $I$  risk-neutral bidders.  I denote by $\mathcal{I} = \{1,2,\cdots, I\}$ the set of bidders and by $\mathcal{J} = \{1,2,\cdots, J\}$ the set of goods. Bidder $i$'s value  of the good $j$  is denoted by $v_i^j$, and bidder $i$'s value vector for all goods is denoted by $\mathbf{v_i}=(v_i^1,v_i^2,\cdots,v_i^J)$. The value profile across bidders is denoted by $\mathbf{v}=\mathbf{(v_1,v_2,\cdots,v_I)}$. Each bidder's value vector is his private information, which the designer perceives as uncertain.  I assume that the designer only knows an upper bound $\bar{v_i}^j$ on the value $v_i^j$ for all $i\in\mathcal{I}$ and all $j\in \mathcal{J}$. Then I denote by $V_i=\times_{j\in\mathcal{J}}[0,\bar{v_i}^j]$ the set of all possible value vectors of bidder $i$, by $V=\times_{i\in \mathcal{I}}V_i$ the set of all possible value profiles across bidders and by $\Delta V$ the set of all possible joint distributions on $V$. In contrast, bidders share a common prior $\pi\in \Delta(V)$.  For exposition, I assume that the supply cost for each good is zero\footnote{All results can be easily extended to the case where the supply cost can be positive and different for each good. Formal statement and proofs are omitted but available upon request.}.  \\
\indent A \textit{mechanism} $\mathcal{M}$ consists of measurable sets of messages $M_i$ for each $i$ and measurable allocation rules $\mathbf{q_i}=(q_i^j)_{j\in\mathcal{J}}:M\to [0,1]^J$ and measurable payment rules: $t_i:M\to \mathbb{R}$ for each $i$, where $M=\times_{i=1}^IM_i$ is the set of message profiles, such that $\sum_{i=1}^Iq_i^j(\mathbf{m})\le 1$ for each $j$. Given a mechanism $\mathcal{M}$ and a simultaneously submitted message profile $\mathbf{m}$, bidder $i$ with a value vector of $\mathbf{v_i}$ has an ex-post payoff \[
U_i(\mathbf{v_i,m})=\mathbf{v_i}\cdot\mathbf{q_i(m)}-t_i(\mathbf{m}).\tag{1}\label{1}\]
Bidders' preferences are quasilinear and additively separable across the goods. I require the mechanism to satisfy a   \textit{participation security} constraint: For each $i$, there exists $\mathbf{0}\in M_i$ such that for each $\mathbf{v_i}\in V_i$ and each $\mathbf{m_{-i}}\in M_{-i}$, \[U_i(\mathbf{v_i,(0,m_{-i})})\ge 0.\tag{PS}\label{ps}\] Bidder $i$ with a value vector $\mathbf{v_i}$ can guarantee a nonnegative ex-post payoff by sending this message, regardless of messages sent by the other bidders.\\
\indent Given a mechanism $\mathcal{M}$ and a joint distribution (common prior among bidders) $\pi$, I have a game of incomplete information. A \textit{Bayes Nash Equilibrium} (BNE) of the game is a strategy profile $\sigma=(\sigma_i)$, $\sigma_i:V_i\to \Delta(M_i)$, such that $\sigma_i$ is best response to $\sigma_{-i}$: Let $U_i(\mathbf{v_i},\mathcal{M},\mathbf{\pi}, \sigma)=\int_{\mathbf{v_{-i}}}U_i(\mathbf{v_i}, (\sigma_i(\mathbf{v_i}),\sigma_{-i}(\mathbf{v_{-i}})))d\mathbf{\pi(v_{-i}|v_i)}$ where $U_i(\mathbf{v_i}, (\sigma_i(\mathbf{v_i}),\sigma_{-i}(\mathbf{v_{-i}})))$ is the multilinear extension of $U_i$ in Equation \eqref{1}, then for any $i,\mathbf{v_i},\sigma_i'$,
    \[U_i(\mathbf{v_i}, \mathcal{M},\pi, \sigma)\ge U_i(\mathbf{v_i}, \mathcal{M},\pi, (\sigma_i',\sigma_{-i})).\tag{BR}\label{br} \]
The set of  all Bayes Nash Equilibria for a given mechanism $\mathcal{M}$ and a given joint distribution $\mathbf{\pi}$ is denoted by $\Sigma(\mathcal{M},\mathbf{\pi})$. \\
\indent The designer's \textit{expected regret} is defined as the difference between the full surplus  given a joint distribution  and the expected revenue under a mechanism and an equilibrium. Given a joint distribution $\mathbf{\pi}$, the full surplus  is $\int_{\mathbf{v}}\{ \sum_{j=1}^J\max_{i\in\mathcal{I}}v_i^j\}d\mathbf{\pi(v)}$, and the expected revenue given a mechanism $\mathcal{M}$ and an equilibrium $\sigma$ is $\int_{\mathbf{v}}\{\sum_{i=1}^I t_i(\sigma(\mathbf{v}))\}d\mathbf{\pi(v)}$. The expected regret thus is $ER(\mathcal{M},\mathbf{\pi}, \sigma)=\int_{\mathbf{v}}\{ \sum_{j=1}^J\max_{i\in\mathcal{I}}v_i^j-\sum_{i=1}^I t_i(\sigma(\mathbf{v}))\}d\mathbf{\pi(v)}$. The integrand is defined as  the \textit{ex-post regret} from $\mathbf{v}$ under the equilibrium $\sigma$. The designer evaluates a mechanism by its worst-case expected regret across all possible joint distributions and equilibria. Formally, the designer evaluates a mechanism $\mathcal{M}$ by $GER(\mathcal{M})=\sup_{\mathbf{\pi}\in \Delta(V)}\sup_{\sigma\in\Sigma(\mathcal{M},\mathbf{\pi})}ER(\mathcal{M},\mathbf{\pi}, \sigma)$,  referred to as the \textit{regret cap}. I say $\bar{R}$ is an \textit{upper bound} of the expected regret under a mechanism $\mathcal{M}$ if $GER(\mathcal{M})\le \bar{R}$. I say $\underline{R}$ is a \textit{lower bound} of the expected regret given a joint distribution $\mathbf{\pi}$ if $\inf_{\mathcal{M}}\inf_{\sigma\in\Sigma(\mathcal{M},\mathbf{\pi})}ER(\mathcal{M},\mathbf{\pi},\sigma)\ge \underline{R}$. The designer's goal is to find a mechanism with the minimal regret cap. I refer to the minimal regret cap as the \textit{minimax regret}. Formally, the designer aims to find a mechanism $\mathcal{M^*}$, referred to as a \textit{minimax regret mechanism}, that solves the following problem:\[
\inf_\mathcal{M}GER(\mathcal{M}).\tag{MRM}\label{mrm}\]  
\section{Methodology}\label{s4}
The problem (\ref{mrm}) can be interpreted as a two-player sequential game. The two players are the designer and adversarial nature. The designer first chooses
a mechanism $\mathcal{M}$. After observing the designer’s choice of
the mechanism, adversarial nature chooses a joint distribution over value profiles  $\mathbf{\pi}\in \Delta(V)$ as well as an equilibrium $\sigma\in \Sigma(\mathcal{M},\mathbf{\pi})$  to maximize the expected regret. The designer's payoff is $-ER(\mathcal{M},\mathbf{\pi}, \sigma)$, and nature's payoff is $ER(\mathcal{M},\mathbf{\pi}, \sigma)$ if $\Sigma(\mathcal{M},\mathbf{\pi})\neq \emptyset$; otherwise, both players' payoffs are minus infinity.   One can also consider a game in which nature moves first by choosing a joint distribution and the designer moves next by choosing a mechanism and an equilibrium. Although it is not obvious that these two problems are payoff equivalent because the equilibrium correspondence is not lower-hemicontinuous, I will construct a mechanism $\mathcal{M^*}$ and a joint distribution $\mathbf{\pi^*}$ that form (a version of) a saddle point and therefore (a version of) the minimax theorem holds.  Precisely, I will (i) find
  an upper bound  $R^*$   of the expected regret under the mechanism $\mathcal{M^*}$ and (ii) show that  $R^*$ is also a lower bound of the expected regret given the joint distribution $\mathbf{\pi^*}$.    Note that (ii) implies that no mechanism can achieve a regret cap strictly lower than $R^*$, and (i) says that the regret cap of  $\mathcal{M^*}$ is weakly lower than $R^*$. Therefore, (i) and (ii) together imply that $\mathcal{M^*}$ is a minimax regret mechanism.  I refer to $\mathbf{\pi^*}$ as a \textit{worst-case  distribution}. 

\section{Main Result}\label{s5}
In this section, I first formally define the separate second-price auction with random reserves $(M^*, \mathbf{q^{*},t^{*}})$ (Section \ref{s51}) and the  joint distribution over value profiles $\mathbf{\pi^*}$ (Section \ref{s52}), then I present the formal statement of the result (Section \ref{s53}) that the proposed mechanism (resp, the proposed distribution) is a minimax regret mechanism (resp, a worst-case  distribution). Finally,  I prove the formal statement (Section \ref{s54}).
\subsection{Separate Second-price Auction with Random Reserves}\label{s51}
The separate second-price auction with random reserves,  $(M^*, \mathbf{q^{*},t^{*}})$,  is defined as follows. First, it is a direct mechanism, i.e., $M^{*}=V$. With slight abuse of notations, I use $\mathbf{v}=(v^j_i)_{i\in\mathcal{I},j\in\mathcal{J}}$ to also denote the reported message. Let $v_{(2)}^j$ be the second highest reported value for the good $j$ whenever well-defined. \\ \indent If there are no ties, $\mathbf{q^{*}(v)}=(q_i^{j*}(\mathbf{v}))_{i\in\mathcal{I},j\in\mathcal{J}}$ where \[q_i^{j*}(\mathbf{v})=\left\{
\begin{array}{lll}
1+\ln{\frac{v_i^j}{\bar{v_i}^j}}    &      & {\text{if $\forall k\neq i, v_i^j>v^j_k$ and $\frac{\bar{v_i}^j}{e}\le v_i^j\le \bar{v_i}^j$};}\\
0    &      & {\text{if $\exists k\neq i$ s.t. $v_i^j<v_k^j$ or $0\le v_i^j<\frac{\bar{v_i}^j}{e}$}.}
\end{array} \right.\]
$\mathbf{t^{*}(v)}=(t_i^{*}(\mathbf{v}))_{i\in\mathcal{I}}$ in which $t_i^{*}(\mathbf{v})=\sum_{j\in\mathcal{J}}t_i^{j*}(\mathbf{v})$\footnote{$t_i^{j*}$ can be interpreted as the payment from bidder $i$ for good $j$ under the mechanism $(M^*, \mathbf{q^{*},t^{*}})$. } where
\[t_i^{j*}(\mathbf{v})=\left\{
\begin{array}{lll}
v_i^j-\frac{\bar{v_i}^j}{e}    &      & {\text{if $\forall k\neq i$, $v_i^j>v^j_k$ and $v^j_{(2)}< \frac{\bar{v_i}^j}{e}\le v_i^j\le \bar{v_i}^j$};}\\
v_i^j+v^j_{(2)}\ln{\frac{v^j_{(2)}}{\bar{v_i}^j}}    &      & {\text{if $\forall k\neq i$, $v_i^j>v^j_k$ and $\frac{\bar{v_i}^j}{e}\le v^j_{(2)}< v_i^j\le \bar{v_i}^j$};}\\
0    &      & {\text{if $\exists k\neq i$ s.t. $v_i^j<v_k^j$ or $0\le v_i^j<\frac{\bar{v_i}^j}{e}$}.}
\end{array} \right.\]
\indent  Now I specify the tie breaking rule when there are ties. Given a value profile $\mathbf{v}$ in which there are ties in the auction of good $j$, let $\mathcal{I}(\mathbf{v^j}):=\{s\in\mathcal{I}|v_s^j\ge v_i^j\quad \forall i\in\mathcal{I}\quad\text{and}\quad v_s\ge \frac{\bar{v_s}^j}{e}\}$. If $\mathcal{I}(\mathbf{v^j})$ is empty, then $q_i^{j*}(\mathbf{v})=t_i^{j*}(\mathbf{v})=0$ for any $i\in\mathcal{I}$; otherwise, pick a bidder $i\in\argmin_{s\in \mathcal{I}(\mathbf{v^j})}\bar{v_s}^j$,  let $q_i^{j*}(\mathbf{v})=1+\ln\frac{v_i^j}{\bar{v_i}^j}, t_i^{j*}(\mathbf{v})=v_i^j+v_i^j\ln{\frac{v_i^j}{\bar{v_i}^j}}$, and  $q_s^{j*}(\mathbf{v})=t_s^{j*}(\mathbf{v})=0$ for any $s\neq i$. In words, among the bidders whose values are weakly higher than their lower bounds of the random reserves respectively, pick a bidder whose upper bound of the values is the lowest and allocate  good $j$ to this bidder when his bid is higher than the random reserve.\\
\indent Note that the allocation probabilities for each good $j$ are independent of the bidders' values for the other goods. The payment rule is characterized by the envelope theorem. Then clearly the above mechanism is equivalent to holding a separate second-price auction with bidder-specific random reserves for  each good.

\subsection{Joint Distribution}\label{s52}
The  joint distribution over value profiles $\mathbf{\pi^*}$ is defined via the following five steps.\\
\textit{\textbf{Step 1: Selection}}. For each good $j$, pick (breaking ties arbitrarily) a bidder $i\in \argmax_{s\in\mathcal{I}}\bar{v_s}^j$. Let $\mathcal{J}(i)$ denote the set of goods for which $i$ is picked. Note that $\mathcal{J}(i)$  (could be empty) is disjoint and  $\cup_{i\in\mathcal{I}}\mathcal{J}(i)=\mathcal{J}$. If a bidder is not picked for any good, then his values for all goods are zeros.\\
\textit{\textbf{Step 2: Equal-revenue Distributions}}. For each bidder $i$ and $j\in\mathcal{J}(i)$ (this part is irrelevant if $\mathcal{J}(i)$ is empty),  the marginal distribution of $v_i^j$ is an equal-revenue distribution whose cumulative distribution function is 
\[\pi_i^{j*}(v_i^j)=\left\{
\begin{array}{lll}
1-\frac{\bar{v_i}^j}{ev_i^j}    &      & {\text{if $\frac{\bar{v_i}^j}{e}\le v_i^j< \bar{v_i}^j$};}\\
1    &      & {\text{if $v_i^j=\bar{v_i}^j$}.}
\end{array} \right.\]
\textit{\textbf{Step 3: Comonotonicity}}. 
For each bidder $i$ and cross $j\in\mathcal{J}(i)$ (this part is irrelevant if $\mathcal{J}(i)$ is empty), the dependence structure is comonotonic. Formally, for good $j\in\mathcal{J}(i)$, define the inverse quantile function \[v_i^j(z_i)=\min\{\tilde{v_i}^j|\pi_i^{j*}(\tilde{v_i}^j)\ge z_i\}=\left\{
\begin{array}{lll}
\frac{\bar{v_i}^j}{e(1-z_i)}    &      & {\text{if $0\le z_i<1-\frac{1}{e}$};}\\
\bar{v_i}^j    &      & {\text{if $z_i\ge 1-\frac{1}{e}$}.}
\end{array} \right.\] Then I define the joint distribution across $j\in\mathcal{J}(i)$  by randomly drawing $z_i\sim U[0,1]$ and taking $v_i^j=v_i^j(z_i)$ for each $j\in\mathcal{J}(i)$.\\
\textit{\textbf{Step 4: Zeros}}. For each bidder $i$ and $j\notin\mathcal{J}(i)$ (this part is irrelevant if $\mathcal{J}(i)=\mathcal{J}$), $v_i^j=0$.\\
\textit{\textbf{Step 5: Independence}}.  For goods across bidders (this part is irrelevant if $\mathcal{J}(i)=\mathcal{J}$ for some $i\in\mathcal{I}$), the values are independently distributed. Formally, $z_i$'s are independently distributed uniform distributions. 
\subsection{Formal Statement: Theorem \ref{t3}}\label{s53}
\begin{theorem}\label{t3}
The mechanism  $(M^*, \mathbf{q^{*},t^{*}})$ is a minimax regret mechanism with the regret cap of $\sum_{j\in\mathcal{J}}\max_{i\in\mathcal{I}}\frac{\bar{v_i}^j}{e}$. The joint distribution $\mathbf{\pi^*}$ is a worst-case  distribution.
\end{theorem}
\subsection{Proof of Theorem \ref{t3}}\label{s54}
\subsubsection{Upper Bound on Regret for $(M^*, \mathbf{q^{*},t^{*}})$}
\begin{proposition}\label{p1}
An upper bound of the expected regret under the mechanism $(M^*, \mathbf{q^{*},t^{*}})$ is $\sum_{j\in\mathcal{J}}\max_{i\in\mathcal{I}}\frac{\bar{v_i}^j}{e}$.
\end{proposition}
\begin{proof}
Consider the auction of good $j$ under the mechanism $(M^*, \mathbf{q^{*},t^{*}})$. Fix an arbitrary joint distribution $\mathbf{\pi}\in \Delta V$.  First consider the truth-telling equilibrium.  Given any value profile $\mathbf{v}\in supp(\mathbf{\pi})$,  suppose $i$ is the unique highest bidder for  good $j$ . If $v^j_{(2)}< \frac{\bar{v_i}^j}{e}\le v_i^j\le \bar{v_i}^j$, then the ex-post regret is $v_i^1-t_i^{j*}(\mathbf{v})=v_i^j-(v_i^j-\frac{\bar{v_i}^j}{e})=\frac{\bar{v_i}^j}{e}$; if $ \frac{\bar{v_i}^j}{e}\le v^j_{(2)}< v_i^j\le \bar{v_i}^j$, then the ex-post regret is $v_i^j-t_i^{j*}(\mathbf{v})=v_i^j-(v_i^j+v^j_{(2)}\ln{\frac{v^j_{(2)}}{\bar{v_i}^j}})=-v^j_{(2)}\ln{\frac{v^j_{(2)}}{\bar{v_i}^j}}$, which is maximized at $v^j_{(2)}=\frac{\bar{v_i}^j}{e}$, yielding an ex-post regret of $\frac{\bar{v_i}^j}{e}$; finally, if $0\le v_i^j<\frac{\bar{v_i}^j}{e}$, then  the ex-post regret is less than $\frac{\bar{v_i}^j}{e}$ because the maximal revenue is less than $\frac{\bar{v_i}^j}{e}$.  Suppose now there are ties. Then under the specified tie breaking rule, if $\mathcal{I}(\mathbf{v^j})$ is empty, then $v^j_i<\max_{i\in\mathcal{I}}\frac{\bar{v_i}^j}{e}$ for any $i\in\mathcal{I}$, and so the ex-post regret is less than $\max_{i\in\mathcal{I}}\frac{\bar{v_i}^j}{e}$; if $\mathcal{I}(\mathbf{v^j})$ is not empty, then the ex-post regret is $v_i^j\ln{\frac{v_i^j}{\bar{v_i}^j}}$ where $i\in\argmin_{s\in \mathcal{I}(\mathbf{v^j})}\bar{v_s}^j$, which is maximized at $v_i^j=\frac{\bar{v_i}^j}{e}$, yielding an ex-post regret of $\frac{\bar{v_i}^j}{e}\le \max_{i\in\mathcal{I}}\frac{\bar{v_i}^j}{e}$. Thus, in the truth-telling equilibrium, the ex-post regret from good $j$ is at most $\max_{i\in\mathcal{I}}\frac{\bar{v_i}^j}{e}$ for any value profile. \\
\indent Next consider  any equilibrium $\sigma$.  Given any value profile $\mathbf{v}\in supp(\mathbf{\pi})$, pick a bidder $i$ whose value for good $j$ is the highest among the bidders, or $i\in \argmax_{i\in I}v_i^j$.  If $v_i^j>\frac{\bar{v_i}^j}{e}$ and there is a positive measure of the others' reports under the conditional equilibrium   report distribution $\sigma_{-i}(\mathbf{v_{-i}})$ such that bidder $i$ wins the good by truthfully reporting $v_i^j$ provided that $v_i^j$ is higher than the random reserve, then bidder $i$ has a strict incentive to truthfully report his  value for good $j$, and thus the  argument in the previous paragraph implies that  the ex-post regret must not exceed  $ \max_{i\in\mathcal{I}}\frac{\bar{v_i}^j}{e}$. Otherwise, there are two cases to consider. 1) If $v_i^j\le\frac{\bar{v_i}^j}{e}$, then the most to lose does  not exceed $\frac{\bar{v_i}^j}{e}$ and thus the ex-post regret must not exceed $\max_{i\in\mathcal{I}}\frac{\bar{v_i}^j}{e}$. 2) If $v_i^j>\frac{\bar{v_i}^j}{e}$ and there is a zero measure of the others' reports under the   conditional equilibrium  report distribution $\sigma_{-i}(\mathbf{v_{-i}})$ such that bidder $i$ wins the good by truthfully reporting $v_i^j$ provided that $v_i^j$ is higher than the random reserve, then (almost surely) the highest report among the other bidders is (weakly) higher than bidder $i$'s value for good $j$. In this case,   the ex-post regret must not exceed $\max_{i\in\mathcal{I}}\frac{\bar{v_i}^j}{e}$ as  the difference between the highest report and the ex-post revenue  is weakly less than $\max_{i\in\mathcal{I}}\frac{\bar{v_i}^j}{e}$ by the  argument in the previous paragraph. Thus, in any equilibrium,   the ex-post regret from good $j$ is at most $\max_{i\in\mathcal{I}}\frac{\bar{v_i}^j}{e}$ for any value profile. This implies that in any equilibrium, the expected regret from good $j$ is at most $\max_{i\in\mathcal{I}}\frac{\bar{v_i}^j}{e}$ given the arbitrary joint distribution $\mathbf{\pi}$. \\
\indent Finally, because of  the separable nature of the mechanism $(M^*, \mathbf{q^{*},t^{*}})$,  an upper bound of the expected regret is $\sum_{j\in\mathcal{J}}\max_{i\in\mathcal{I}}\frac{\bar{v_i}^j}{e}$. 
\end{proof}
\begin{remark}
\normalfont This upper bound is hit given the joint distribution $\mathbf{\pi^*}$ and the truth-telling equilibrium. To see this, fix any $\mathbf{v}\in \mathbf{\pi^*}$ and  consider good $j$. By the definition of  $\mathbf{\pi^*}$, there is only one bidder, denoted by $i$,  whose value for good $j$ is positive. In addition, $v_i^j\ge \frac{\bar{v_i}^j}{e}$. Then by the proof of Proposition \ref{p1}, the ex-post regret from good $j$ is $\frac{\bar{v_i}^j}{e}$ in the truth-telling equilibrium. Note that by the definition of  $\mathbf{\pi^*}$, $i\in \argmax_{i\in \mathcal{I}}\bar{v_i}^j$. Thus the ex-post regret from good $j$ is equal to  $\max_{i\in\mathcal{I}}\frac{\bar{v_i}^j}{e}$ in the truth-telling equilibrium. Because this is true for any $\mathbf{v}\in \mathbf{\pi^*}$, the expected regret from good $j$ is equal to  $\max_{i\in\mathcal{I}}\frac{\bar{v_i}^j}{e}$ given the joint distribution $\mathbf{\pi^*}$ and the truth-telling equilibrium. Summing up across goods, the expected regret is  $\sum_{j\in\mathcal{J}}\max_{i\in\mathcal{I}}\frac{\bar{v_i}^j}{e}$ given the joint distribution $\mathbf{\pi^*}$ and the truth-telling equilibrium. 
\end{remark}
\begin{remark}\label{r2}
\normalfont \cite{koccyiugit2020regret} present a separate second-price auction with anonymous random reserves   in which \[q_i^{j}(\mathbf{v})=\left\{
\begin{array}{lll}
1+\ln{\frac{v_i^j}{\max_{i\in\mathcal{I}}\bar{v_i}^j}}    &      & {\text{if $\forall k\neq i$, $v_i^j>v^j_k$ and $\frac{\max_{i\in\mathcal{I}}\bar{v_i}^j}{e}\le v_i^j\le \bar{v_i}^j$};}\\
0    &      & {\text{if $ \exists k\neq i$ s.t. $v_i^j<v_k^j$ or $0\le v_i^j<\frac{\max_{i\in\mathcal{I}}\bar{v_i}^j}{e}$}.}
\end{array} \right.\] And
the payment rule is characterized by the envelope theorem. They show that given this mechanism,  $\sum_{j\in\mathcal{J}}\max_{i\in\mathcal{I}}\frac{\bar{v_i}^j}{e}$ is an upper bound of the ex-post regret.  Although their mechanism has the same regret cap, the designer may favor $(M^*, \mathbf{q^{*},t^{*}})$  over their mechanism for reasons outside the model. Specifically,  there exists a Pareto ranking between the two mechanisms in the following sense: as long as $\max_{i\in\mathcal{I}}\bar{v_i}^j>\bar{v_k}^j$ for some $k\in\mathcal{I}$ and some $j\in\mathcal{J}$,     it is straightforward to show that in the truth-telling equilibrium,  i) the ex-post regret under $(M^*, \mathbf{q^{*},t^{*}})$  is weakly lower than that  under their mechanism for any $\mathbf{v}\in V$, and ii) the ex-post regret under $(M^*, \mathbf{q^{*},t^{*}})$  is strictly lower than that under their mechanism for some $\mathbf{v}\in V$. Intuitively, under their mechanism,  the allocation probability is lower for a highest bidder whose upper bound of the values of the good is not the highest, resulting in a higher ex-post regret for such a value profile.
\end{remark}
\indent Indeed, under the criterion used in Remark \ref{r2}, no   mechanism from the family of separate second-price auctions\footnote{In a separate second-price auction, there may be random reserves, deterministic reserves, or no reserves.}, denoted by $\mathcal{F}-SSP$,   is better than the mechanism $(M^*, \mathbf{q^{*},t^{*}})$.
\begin{definition}\label{d1}
\normalfont I say a direct mechanism $\mathcal{M}=(V,\mathbf{q,t})$ is a \textit{dominant-strategy mechanism} if for all $i\in \mathcal{I}$, all $\mathbf{v}\in V$, and all $\mathbf{v_i'}\in V_i$,
\[\mathbf{v_i}\cdot \mathbf{q_i(v)}-t_i(\mathbf{v})\ge \mathbf{v_i}\cdot \mathbf{q_i(v_i',v_{-i})}-t_i(\mathbf{v_i',v_{-i}}), \]\[\mathbf{v_i}\cdot \mathbf{q_i(v)}-t_i(\mathbf{v})\ge 0.\]
\end{definition}
\begin{definition}
\normalfont I say a dominant-strategy mechanism $\mathcal{M}_1$ is \textit{undominated} by another dominant-strategy mechanism $\mathcal{M}_2$ if in the truth-telling equilibrium, the ex-post regret under the mechanism $\mathcal{M}_1$ is strictly lower than that under the mechanism $\mathcal{M}_2$ for some $\mathbf{v}\in V$.
\end{definition}
\begin{corollary}
The mechanism $(M^*, \mathbf{q^{*},t^{*}})$ is undominated by any mechanism from $\mathcal{F}-SSP$.
\end{corollary}
\begin{proof}
Fix any $i\in \mathcal{I}$ and any $j\in \mathcal{J}$, consider the value profiles in which $v_i^j\in [0,\bar{v_i}^j]$ and all other values are zeros. Then Proposition 1 in \cite{bergemann2008pricing} implies that the random reserve for  the bidder $i$ and the good $j$ in the mechanism $(M^*, \mathbf{q^{*},t^{*}})$  is the unique random reserve that minimizes the worst-case ex-post regret in the truth-telling equilibrium for these value profiles. Therefore, if a different (random) reserve were used for the bidder $i$ and the good $j$, then there would be a value profile with an ex-post regret strictly  higher than that under the mechanism $(M^*, \mathbf{q^{*},t^{*}})$. \\
\indent In addition, the specific tie-breaking rule  in $(M^*, \mathbf{q^{*},t^{*}})$ minimizes the worst-case ex-post regret when there are ties across different tie-breaking rules. To see this, recall  that  the worst-case ex-post regret is proportional to the upper bound of values in the  one-good one-bidder case and that  under the  tie-breaking rule in $(M^*,  \mathbf{q^{*},t^{*}})$, a bidder with the lowest upper bound of values for a good is picked. This finishes the proof.
\end{proof}
\subsubsection{Lower Bound on Regret for $\mathbf{\pi^*}$}
\begin{proposition}\label{p2}
A lower bound of the expected regret under $\mathbf{\pi^*}$ is $\sum_{j\in\mathcal{J}}\max_{i\in\mathcal{I}}\frac{\bar{v_i}^j}{e}$.
\end{proposition}
\begin{proof}
Note that given a joint distribution, minimizing the expected regret across mechanisms and equilibria is equivalent to maximizing the expected revenue  across mechanisms and equilibria. Then the revelation principle applies and thus it is without loss to restrict attention to direct mechanisms.\\
\indent I parameterize the value profile across bidders 
by $\mathbf{z}=(z_1,z_2,\cdots,z_{I})\in [0,1]^I$.  Then for any direct mechanism $(\mathbf{q(z),t(z)})=((\mathbf{q_i(z)})_{i\in \mathcal{I}},(t_i(\mathbf{z}))_{i\in \mathcal{I}})$ where $\mathbf{q_i(z)}=(q_i^j(\mathbf{z}))_{j\in\mathcal{J}}\in [0,1]^J$ represent the allocation probabilities of the goods  to bidder $i$  under the parametrerized value profile $\mathbf{z}$ and $t_i(\mathbf{z})\in \mathbb{R}$ represents bidder $i$'s payment under $\mathbf{z}$,  (\ref{br}) together with (\ref{ps}) imply \[U_i(z_i):=\sum_{j\in \mathcal{J}(i)}v_i^j(z_i)Q_i^j(z_i)-T_i(z_i)\ge \sum_{j\in \mathcal{J}(i)}v_i^j(z_i)Q_i^j(z_i')-T_i(z_i') \quad \text{for}\quad i\in\mathcal{I},z_i,z_i'\in[0,1],\tag{BIC}\label{bic}\]
\[\sum_{j\in \mathcal{J}(i)}v_i^j(z_i)Q_i^j(z_i)-T_i(z_i)\ge 0 \quad \text{for}\quad i\in\mathcal{I},z_i\in [0,1],\tag{BIR}\label{bir}\]
where $Q_i^j(z_i)=\int_{[0,1]^{I-1}}q_i^j(z_i,\mathbf{z_{-i}})d\mathbf{z_{-i}}$ and $T_i(z_i)=\int_{[0,1]^{I-1}}t_i(z_i,\mathbf{z_{-i}})d\mathbf{z_{-i}}$ are the expected allocation of good $j$ to type $z_i$ of bidder $i$ and the expected payment made by type $z_i$ of bidder $i$ respectively, due to the fact that $z_i$'s are independently distributed uniform distributions by the definition of $\mathbf{\pi^*}$.
Note that  the allocation of good $j$ for $j\notin \mathcal{J}(i)$ does not appear in either (\ref{bic}) or (\ref{bir}) because the value for such a  good (if any) is zero to bidder $i$ under the joint distribution $\mathbf{\pi^*}$. \\
\indent For $z_i'\ge z_i$,  (\ref{bic}) implies that  \[
\sum_{j\in \mathcal{J}(i)}(v_i^j(z_i')-v_i^j(z_i))Q_i^j(z_i')\ge U_i(z_i')-U_i(z_i)\ge \sum_{j\in \mathcal{J}(i)}(v_i^j(z_i')-v_i^j(z_i))Q_i^j(z_i).\tag{2}\label{5.3.1}\]
Then  $U_i(z_i)$ is Lipschitz, thus absolutely continuous w.r.t. $z_i$, and so equal to the integral of its
derivative. In addition, note that $v_i^j(z_i)$ is differentiable for all $z_i$ but $z_i=1-\frac{1}{e}$. Then applying the envelope theorem to \eqref{5.3.1} at each point of differentiability, I obtain that 
\[\frac{\partial U_i(z_i)}{\partial z_i}=\sum_{j\in\mathcal{J}(i)}\frac{\partial v_i^j(z_i)}{\partial z_i}Q_i^j(z_i)=\left\{
\begin{array}{lll}
\sum_{j\in\mathcal{J}(i)}\frac{\bar{v_i}^j}{e(1-z_i)^2}Q_i^j(z_i)    &      & {\text{if $0\le z_i<1-\frac{1}{e}$};}\\
0    &      & {\text{if $z_i> 1-\frac{1}{e}$}.}
\end{array} \right.\]
Thus, \[U_i(z_i)=\left\{
\begin{array}{lll}
U_i(0)+\int_0^{z_i}[ \sum_{j\in\mathcal{J}(i)}\frac{\bar{v_i}^j}{e(1-\tilde{z_i})^2}Q_i^j(\tilde{z_i})] d\tilde{z_i}   &      & {\text{if $0\le z_i<1-\frac{1}{e}$};}\\
U_i(0)+\int_0^{1-\frac{1}{e}}[ \sum_{j\in\mathcal{J}(i)}\frac{\bar{v_i}^j}{e(1-\tilde{z_i})^2}Q_i^j(\tilde{z_i})] d\tilde{z_i}  &      & {\text{if $z_i\ge 1-\frac{1}{e}$}.}
\end{array} \right.\]
Therefore,   the expected revenue from bidder $i$ 
\begin{equation*} 
\begin{split}
\int_0^1T_i(z_i)dz_i & = \int_0^1[\sum_{j\in\mathcal{J}(i)}v_i^j(z_i)Q_i^j(z_i)-U_i(z_i)]dz_i \\
 & = \int_0^{1-\frac{1}{e}}\{\sum_{j\in\mathcal{J}(i)}v_i^j(z_i)Q_i^j(z_i)-U_i(0)-\int_0^{z_i}[ \sum_{j\in\mathcal{J}(i)}\frac{\bar{v_i}^j}{e(1-\tilde{z_i})^2}Q_i^j(\tilde{z_i})] d\tilde{z_i}\}dz_i+\\
 &\int_{1-\frac{1}{e}}^1\{\sum_{j\in\mathcal{J}(i)}v_i^j(z_i)Q_i^j(z_i)-U_i(0)-\int_0^{1-\frac{1}{e}}[ \sum_{j\in\mathcal{J}(i)}\frac{\bar{v_i}^j}{e(1-\tilde{z_i})^2}Q_i^j(\tilde{z_i})] d\tilde{z_i}\}dz_i\\
 &\le \int_0^{1-\frac{1}{e}}\{\sum_{j\in\mathcal{J}(i)}v_i^j(z_i)Q_i^j(z_i)-\int_0^{z_i}[ \sum_{j\in\mathcal{J}(i)}\frac{\bar{v_1}^j}{e(1-\tilde{z_i})^2}Q_i^j(\tilde{z_i})] d\tilde{z_i}\}dz_i+\\
 &\int_{1-\frac{1}{e}}^1\{\sum_{j\in\mathcal{J}(i)}v_i^j(z_i)Q_i^j(z_i)-\int_0^{1-\frac{1}{e}}[ \sum_{j\in\mathcal{J}(i)}\frac{\bar{v_i}^j}{e(1-\tilde{z_i})^2}Q_i^j(\tilde{z_i})] d\tilde{z_i}\}dz_i\\
 &=\sum_{j\in \mathcal{J}(i)}\{\int_0^{1-\frac{1}{e}}[(v_i^j(z_i)-(1-\frac{1}{e}-z_i)\frac{\bar{v_i}^j}{e(1-z_i)^2})Q_i^j(z_i)]dz_i+\\
 &\int_{1-\frac{1}{e}}^1[v_i^j(z_i)Q_i^j(z_i)-\int_0^{1-\frac{1}{e}}[\frac{\bar{v_i}^j}{e(1-\tilde{z_i})^2}Q_i^j(\tilde{z_i})] d\tilde{z_i}]dz_i\}\\
 &=\sum_{j\in \mathcal{J}(i)}\{\int_0^{1-\frac{1}{e}}[(v_i^j(z_i)-(1-z_i)\frac{\bar{v_i}^j}{e(1-z_i)^2})Q_i^j(z_i)]dz_i+\int_{1-\frac{1}{e}}^1[(v_i^j(z_i)Q_i^j(z_i)]dz_i\}\\
 &=\sum_{j\in \mathcal{J}(i)}\int_{1-\frac{1}{e}}^1[\bar{v_i}^jQ_i^j(z_i)]dz_i\le \sum_{j\in \mathcal{J}(i)}\frac{\bar{v_i}^j}{e},
\end{split}
\end{equation*}
where the first inequality holds because (\ref{bir}) implies that $U_i(0)\ge 0$, the third equality is obtained via integration by parts, the last equality holds because $v_i^j(z_i)-(1-z_i)\frac{\bar{v_i}^j}{e(1-z_i)^2}=0$ for $0\le z_i< 1-\frac{1}{e}$ and $v_i^j(z_i)=\bar{v_i}^j$ for $z_i>1-\frac{1}{e}$, and the last inequality holds because $Q_i^j(z_i)\le 1$.\\
\indent Then, the expected revenue from all the bidders
\begin{equation*}
\begin{split}
 \sum_{i=1}^I\int_0^1T_i(z_i)dz_i & \le \sum_{i\in\mathcal{I}}\sum_{j\in \mathcal{J}(i)} \frac{\bar{v_i}^j}{e}\\
 &=\sum_{j\in\mathcal{J}}\max_{i\in\mathcal{I}}\frac{\bar{v_i}^j}{e},   
\end{split}
\end{equation*} 
where the  equality holds by  the definition of $\mathcal{J}(i)$.\\
\indent Now, the expected regret 
\begin{equation*}
    \begin{split}
       \sum_{i\in\mathcal{I}}\int_0^1\sum_{j\in\mathcal{J}(i)}v_i^j(z_i)dz_i-\sum_{i=1}^I\int_0^1T_i(z_i)dz_i &=\sum_{j\in\mathcal{J}}\max_{i\in\mathcal{I}}\frac{2\bar{v_i}^j}{e}-\sum_{i=1}^I\int_0^1T_i(z_i)dz_i\\&\ge \sum_{j\in\mathcal{J}}\max_{i\in\mathcal{I}}\frac{\bar{v_i}^j}{e},
    \end{split}
\end{equation*}
where the term $\sum_{i\in\mathcal{I}}\int_0^1\sum_{j\in\mathcal{J}(i)}v_i^j(z_i)dz_i$ is the full surplus given the joint distribution $\mathbf{\pi^*}$ and  the equality holds by direct calculation and by the definition of $\mathcal{J}(i)$. Thus, $\sum_{j\in\mathcal{J}}\max_{i\in\mathcal{I}}\frac{\bar{v_i}^j}{e}$ is a lower bound of the expected regret under $\mathbf{\pi^*}$. 
\end{proof}
\begin{remark}\label{r1}
\normalfont In the Step 1 of the definition of $\mathcal{\pi^*}$, it is  important that for each good, only one bidder is selected  when there are ties. Otherwise, there would be competition for some good, resulting in a lower  expected regret.
\end{remark}
\begin{remark}
\normalfont One may be tempted to consider the following joint distribution over value profiles as a candidate for a worst-case  distribution.   There is only one bidder, bidder $i$,    whose values for the goods are non-zero; in addition, the bidder $i$'s values for the goods follow the  comonotonic equal-revenue distribution. The bidder $i$ is selected such that $i\in \argmax_{i\in \mathcal{I}}\sum_{j\in\mathcal{J}}\frac{\bar{v_i}^j}{e}$.  Then by an argument similar to the proof of Proposition \ref{p2} , a lower bound of the expected regret under this joint distribution is  $\max_{i\in\mathcal{I}}\sum_{j\in\mathcal{J}}\frac{\bar{v_i}^j}{e}$. However, this lower bound is lower than   $ \sum_{j\in\mathcal{J}}\max_{i\in\mathcal{I}}\frac{\bar{v_i}^j}{e}$ in general. Thus, this joint distribution is not ``bad'' enough for the designer and  is not a worst-case  distribution in general. Intuitively, this joint distribution may ignore a bidder whose upper bound of the values of a given good  is the highest among the bidders, resulting in an expected regret not high enough. This motivates the Step 1 of  the definition of $\mathbf{\pi^*}$. 
\end{remark}
\begin{remark}\label{r5}
\normalfont What if the designer knows nothing about the joint distribution over bidders' value profiles? That is, bidders' values can be unbounded. I argue that the regret cap for any mechanism that secures bidders' participation will be infinity. To see this, consider a joint distribution that puts all probability masses on  a single value profile in which  bidder $i$ has a large positive value of $\theta$ for good $j$, bidder $i$'s values for the other goods are zeros and the other bidders' values for all the goods are zeros. Recall that given  a joint distribution over value profiles, the revelation principle applies and it is without loss to restrict attention to direct mechanisms. In addition, the expected revenue is generated from selling good $j$ to bidder $i$ only, as the other values are zeros and the mechanism secures bidders' participation. Consider a revenue-maximizing (and therefore regret-minimizing) direct mechanism,  let $Q_i^j(x)$ denote the expected allocation probability of good $j$ to bidder $i$ given a bidder $i$'s report of $x$ about his  value for good $j$. Note that $Q_i^j(x)$ is non-decreasing in $x$ by the incentive compatible constraint.  Then the expected revenue is $T_i^j(\theta)=\theta Q_i^j(\theta)-\int_{0}^{\theta}Q(x)dx$. Define  $\lim_{v_i^j\to \infty}Q_i^j(v_i^j):=\kappa$. By definition,  for any $\epsilon >0$, there exists a $t\ge 0$ such that $Q_i^j(v_i^j)\ge \kappa-\epsilon$ for any $v_i^j\ge t$. Then $\int_{0}^{\theta}Q_i^j(x)dx=\int_{0}^{t}Q_i^j(x)dx+\int_t^{\theta}Q_i^j(x)dx\ge (\theta-t)(\kappa-\epsilon)$, so $T_i^j(\theta)\le  \theta \kappa-(\theta-t)(\kappa-\epsilon)=\epsilon \theta+t(\kappa-\epsilon)$, and the expected regret is $\theta-T_i^j(\theta)\ge (1-\epsilon)\theta-t(\kappa-\epsilon)$. As $\epsilon$ can be chosen to be arbitrarily small, the expected regret goes to infinity as $\theta$ goes to infinity\footnote{This proof is similar to the proof of Proposition 1 in \cite{carrasco2017optimal}.}.  
\end{remark}
\indent Theorem \ref{t3} follows immediately from Proposition \ref{p1} and \ref{p2}.
\section{Special Cases}\label{s6}
\indent In this section, I present the results for  two special cases in which $I=1$ and $J=1$ respectively, which correspond to multi-dimensional screening (Section \ref{s61}) and single-good auction (Section \ref{s62}). 
\subsection{Multi-Dimensional Screening: $I=1$}\label{s61}
Let the mechanism $(M_1^*, \mathbf{q_1^*},t_1^*)$ (resp, the joint distribution $\mathbf{\pi_1^*}$ ) be the specialization of the mechanism   $(M^*, \mathbf{q^{*},t^{*}})$ (resp, the joint distribution $\mathbf{\pi^*}$) to the case in which $I=1$. I omit their descriptions for brevity. Note that the mechanism $(M_1^*, \mathbf{q_1^*},t_1^*)$ is a \textit{separate randomized  posted-price mechanism}: each good is sold separately with a random posted price. In the joint distribution $\mathbf{\pi_1^*}$, the marginal distribution of each good is an equal-revenue distribution and the values across goods are comonotonic.  
\begin{corollary}[Multi-Dimensional Screening]\label{t1}
If $I=1$, then the  mechanism $(M_1^*, \mathbf{q_1^*},t_1^*)$ is a minimax regret mechanism with the regret cap of $\sum_{j\in\mathcal{J}}\frac{\bar{v_1}^j}{e}$.  The joint distribution $\mathbf{\pi_1^*}$ is a worst-case  distribution.

\end{corollary}
\begin{proof}
The proof is a straightforward adaptation of the proof of Theorem \ref{t3} to the case in which $I=1$. 
\end{proof}
\begin{remark}
\normalfont There are very limited results in multi-dimensional screening for other correlation structures. \cite{mcafee1989multiproduct} show that with independent continuous distributions, separate selling is essentially never optimal. Therefore an independent joint distribution,  where the marginal distributions remain the same but the values across the goods are independent, is not a worst-case  distribution. 
\end{remark}
\begin{remark}\label{r7}
\normalfont There is another minimax regret mechanism for the multi-dimensional screening: a randomized grand bundling. It can be described as follows. The designer  sells the bundle of all the goods only.   Let $b$ be the bid for the bundle of all the goods. If $b>\sum_{j\in\mathcal{J}}\frac{\bar{v_1}^j}{e}$, then allocate the bundle with a probability of $1+\ln{\frac{b}{\sum_{j\in\mathcal{J}}\bar{v_1}^j}}$ and charge a  price of $b-\sum_{j\in\mathcal{J}}\frac{\bar{v_1}^j}{e}$; otherwise, no goods are allocated and the buyer (the bidder 1) pays nothing.  It is straightforward to show that the regret cap of this mechanism is $\sum_{j\in\mathcal{J}}\frac{\bar{v_1}^j}{e}$. 
\end{remark}
\subsection{Single-Good Auction: $J=1$}\label{s62}
Let the mechanism $(M^{1*},\mathbf{q^{1*},t^{1*}})$ (resp, the joint distribution $\mathbf{\pi^{1*}}$ ) be the specialization of the mechanism   $(M^*, \mathbf{q^{*},t^{*}})$ (resp, the joint distribution $\mathbf{\pi^*}$) to the case in which $J=1$. I omit their descriptions for brevity. Note that the mechanism $(M^{1*},\mathbf{q^{1*},t^{1*}})$ is a \textit{second-price auction with random reserves}: the single good   is auctioned  via a second-price auction with bidder-specific random reserves. In the joint distribution $\mathbf{\pi^{1*}}$, only the bidder with the highest upper bound of the values for the good has a positive value (breaking ties arbitrarily) and  the marginal distribution of this bidder's value is an equal-revenue distribution. 
\begin{corollary}[Single-Good Auction]\label{t2}
If $J=1$, then the mechanism $(M^{1*},\mathbf{q^{1*},t^{1*}})$ is a minimax regret mechanism with the regret cap of  $\max_{i\in\mathcal{I}}\frac{\bar{v_i}^1}{e}$.  The joint distribution $\mathbf{\pi^{1*}}$ is a worst-case  distribution.

\end{corollary}
\begin{proof}
The proof is a straightforward adaptation of the proof of Theorem \ref{t3} to the case in which $J=1$. 
\end{proof}
\section{Discussion}\label{s7}
\subsection{Solution Concept}
In this paper, I consider  the class of all mechanisms that secure bidders' participation and the worst Bayes Nash Equilibrium for the designer. The solution concept follows from a recent literature on informationally robust mechanism design, e.g., \cite{du2018robust} and \cite{brooks2021optimal}. Several remarks can be made in sequence.  First, if we assume that the class of mechanisms is the set of dominant-strategy  mechanisms  and that the truth-telling equilibrium is played, then the same result will hold by a simple extension of the current proofs. This is because under the constructed worst-case distribution, the expected regret under the best dominant-strategy mechanism is the same as that under the best Bayesian incentive-compatible mechanism.  Second,   for the main result, it is not crucial that  adversarial nature has to pick the worst equilibrium. That is, we can allow adversarial nature to pick the best equilibrium for the designer, and the same result will still hold. So the main result may be a priori surprising result: the class of the mechanisms is much wider than the set of dominant-strategy mechanisms,  yet, a dominant-strategy mechanism emerges as  a minimax regret mechanism.
\subsection{Comparative Statics}
It is instructive to discuss some comparative statics assuming that there is no trivial good or bidder, i.e., $\bar{v_i}^j>0$ for any $i\in \mathcal{I}$ and $j\in \mathcal{J}$. First, the minimax regret is strictly increasing in  $J$. To understand this, note that  the comonotonic structure in the worst-case  distribution  reduces multi-dimensional screening to single-dimensional screening, then  when adding a new good, the minimax regret will  increase by the amount of  the minimax regret when there is only this new good. Second, the minimax regret is weakly increasing in $I$. To understand this, note that the  zero values in the worst-case  distribution eliminate the competition\footnote{The competition would increase with $I$ for general joint distributions. For example,  consider a joint distribution in which  bidders' values of a given good follow  $i.i.d.$ uniform distributions.  It is straightforward to show that  the expected regret under an optimal mechanism would eventually go to 0 as $I$ goes to infinity given this joint distribution.} for a given good, then as the full surplus weakly increases with $I$,  the minimax regret also weakly increases with $I$ (strictly increases with $I$ when the new bidder's upper bound of the values of some good is higher than that of any previous bidder). Third, for the symmetric case where the upper bounds of the values for a given good are the same across bidders, or $\bar{v_i}^j=\bar{v_k}^j$ for any $i\in \mathcal{I}$, any $k\in \mathcal{I}$ and any $j\in \mathcal{J}$, the \textit{average} minimax regret (the minimax regret divided by $I$) is strictly decreasing in $I$. To understand this, note that when adding a symmetric bidder, the full surplus does not change given the worst-case  distribution, and, again, there is still no competition for any good. Then, the minimax regret remains the same and thus the average minimax regret is strictly decreasing in $I$.

\subsection{Digital Goods}
Consider a related problem in which the designer auctions  digital goods\footnote{A digital goods auction is an auction in which the designer has an unlimited supply of the same good.} to $I$ bidders, e.g., e-books, mobile apps, online courses, etc. Each bidder  demands at most one unit of the good.  Bidder $i$ has a private value $v_i\in [0,\bar{v_i}]$. The designer aims to minimize the worst-case expect regret. The formal objective function can be similarly defined.  Indeed,  this problem may be interpreted as a special case of the model: there are $I$ different goods, but each bidder values only one of the goods and the good  each bidder values is different. Under this interpretation, adversarial nature's ability is ``constrained'' in that the set of possible joint distributions is smaller than the previous one.  Note however that the worst-case distribution in Theorem \ref{t3} is not excluded. Then a direct implication of Theorem \ref{t3} is that   a \textit{separate randomized posted-price mechanism} as follows is a minimax regret mechanism for this problem:  \[q_i(v_1,v_2,\cdots,v_I)=\left\{
\begin{array}{lll}
1+\ln{\frac{v_i}{\bar{v_i}}}    &      & {\text{if $\frac{\bar{v_i}}{e}\le v_i\le \bar{v_i}$};}\\
0    &      & {\text{if $0\le v_i<\frac{\bar{v_i}}{e}$}.}
\end{array} \right.\]
And the payment rule is characterized by the envelope theorem. Note that the allocation to bidder $i$ depends on bidder $i$'s value only.  In addition, an \textit{independent equal-revenue distribution} as follows is a worst-case  distribution: the marginal distribution of $v_i$ follows an equal-revenue distribution whose cumulative distribution function is \[\pi_i(v_i)=\left\{
\begin{array}{lll}
1-\frac{\bar{v_i}}{ev_i}    &      & {\text{if $\frac{\bar{v_i}}{e}\le v_i< \bar{v_i}$};}\\
1    &      & {\text{if $v_i=\bar{v_i}$}.}
\end{array} \right.\]
And the values across bidders are independent.
\section{Concluding Remarks}\label{s8}
In this paper, I characterize a simple minimax regret mechanism for auctioning multiple goods given general upper bounds of values. It is worth noting that the proposed  mechanism is strategy-proof. Hence,  it (essentially\footnote{A strategy-proof  mechanism can be slightly perturbed so that truth-telling is the unique equilibrium.}) remains a  minimax regret mechanism even without the assumption of a common prior  among bidders.  Critically,   I drop the extreme assumption made by the traditional mechanism design literature that the designer knows the joint distribution over value profiles, but impose an equally extreme assumption that the designer has no distributional information except for the upper bounds of values, on which the result heavily relies. I believe the truth lies in  intermediate cases, which are interesting to further explore.  I further conjecture that separation remains a property in many other informational environments. 
\appendix
\section{Additional Information}\label{s9}
\indent An \textit{additional information structure} consists of a measuable set of additional information $S_i$ for each bidder $i$, with $S=\times_{i=1}^IS_i$, and a joint distribution $\delta\in \Delta(V\times S)$. An additional information structure is denoted by $\mathcal{T}=(S,\delta)$. I say  $\mathcal{T}$ is $\pi-consistent$ if the marginal of $\delta$ on $V$ is $\pi$, i.e., for every measurable $\tilde{V}\subseteq V$, $\delta(\tilde{V}\times S)=\pi(\tilde{V})$. The set of all $\pi-consistent$ additional information structures is denoted by $\mathbf{T}(\pi)$. As before, each bidder $i$ knows his private value vector $\mathbf{v_i}\in V_i$. And $\pi$ is their common prior.  But,  before playing a mechanism, each bidder $i$ may observe a signal $\mathbf{s_i}\in S_i$ from an additional information structure $\mathcal{T}\in \mathbf{T}(\pi)$. And $\mathcal{T}$ is their common knowledge. The definition of and the requirement for a mechanism are the same as before. Given a mechanism $\mathcal{M}$ and a common prior $\mathbf{\pi}$ and an additional information structure $\mathcal{T}\in \mathbf{T}(\pi)$, I have a game of incomplete information.  With slight abuse of notations,  a \textit{Bayes Nash Equilibrium} (BNE) of the game is a strategy profile $\sigma=(\sigma_i)$, $\sigma_i:V_i\times S_i\to \Delta(M_i)$, such that $\sigma_i$ is best response to $\sigma_{-i}$: Let $U_i(\mathbf{v_i,s_i},\mathcal{M},\mathbf{\pi},\mathcal{T}, \sigma)=\int_{\mathbf{v_{-i},s_{-i}}}U_i(\mathbf{v_i}, (\sigma_i(\mathbf{v_i,s_i}),\sigma_{-i}(\mathbf{v_{-i},s_{-i}})))d\mathbf{\delta(v_{-i},s_{-i}|v_i,s_i)}$ where $U_i(\mathbf{v_i}, (\sigma_i(\mathbf{v_i,s_i}),\sigma_{-i}(\mathbf{v_{-i},s_{-i}})))$ is the multilinear extension of $U_i$ in Equation (\ref{1}), then for any $i,\mathbf{v_i,s_i},\sigma_i'$,
    \[U_i(\mathbf{v_i,s_i},\mathcal{M},\mathbf{\pi},\mathcal{T}, \sigma)\ge U_i(\mathbf{v_i,s_i},\mathcal{M},\mathbf{\pi},\mathcal{T}, (\sigma_i',\sigma_{-i})).\tag{BR'}\label{br'} \]
The set of  all Bayes Nash Equilibria for a given mechanism $\mathcal{M}$ and a given common prior $\mathbf{\pi}$ and a given additional information structure $\mathcal{T}\in\mathbf{T(\pi)}$ is denoted by $\Sigma(\mathcal{M},\mathbf{\pi},\mathcal{T})$.\\
\indent Given a common prior $\mathbf{\pi}$ and an additional information structure $\mathcal{T}\in \mathbf{T(\pi)}$, the expected regret  is $ER'(\mathcal{M},\mathbf{\pi},\mathcal{T},  \sigma)=\int_{\mathbf{v,s}}\{ \sum_{j=1}^J\max_{i\in\mathcal{I}}v_i^j-\sum_{i=1}^I t_i(\sigma(\mathbf{v,s}))\}d\mathbf{\delta(v,s)}$. The designer evaluates a mechanism by its worst-case expected regret across all possible common priors and consistent additional information structures and equilibria. Formally, the designer evaluates a mechanism $\mathcal{M}$ by $GER'(\mathcal{M})=\sup_{\mathbf{\pi}\in \Delta(V)}\sup_{\mathcal{T}\in \mathbf{T(\pi)}}\sup_{\Sigma(\mathcal{M},\mathbf{\pi},\mathcal{T})}ER'(\mathcal{M},\mathbf{\pi}, \mathcal{T}, \sigma)$. The designer's goal is to find a mechanism with the minimal worst-case expected regret. Formally, the designer aims to find a mechanism, referred to as a \textit{min-3max regret mechanism}, that solves the following problem:\[
\inf_\mathcal{M}GER'(\mathcal{M}).\tag{MRM'}\label{mrm'}\]
\begin{manualtheorem}{1'}\label{t1'}
The mechanism  $(M^*, \mathbf{q^{*},t^{*}})$ is a min-3max regret mechanism. 
\end{manualtheorem}
\begin{proof}
For adversarial nature's strategy, let the common prior be $\mathbf{\pi^*}$ and the set of additional information $S$ be a singleton. The proof of Theorem \ref{t3} then applies.
\end{proof}
Intuitively, adversarial nature cannot generate strictly more expected regret even though it can use additional information structures because the mechanism $(M^*,\mathbf{q^*,t^*})$ is strategy-proof.
\bibliographystyle{apalike}
\bibliography{abc}

\end{document}